\documentclass[]{interact}

\usepackage{epstopdf}
\usepackage[caption=false]{subfig}

\usepackage[natbibapa,nodoi]{apacite}
\usepackage{hyperref}
\usepackage{cleveref}

\usepackage{algorithm}
\usepackage{algpseudocode}

\newcommand{\ouset}[3]{\underset{#1}{\overset{#3}{#2}}}

\theoremstyle{plain}
\newtheorem{theorem}{Theorem}[section]

\theoremstyle{definition}

\theoremstyle{remark}

\begin{document}


\title{A more efficient algorithm to compute the Rand Index for change-point problems}

\author{
\name{Lucas de Oliveira Prates\thanks{CONTACT Prates.~L. Email: lucasdelprates at gmail dot com}}
}

\maketitle

\begin{abstract}
    We provide a more efficient algorithm for computing the Rand Index when the data clusters come from a change-point detection problem. Given the number of data points  $N$ and two change-point sets of size $r$ and $s$, the algorithm runs on $O(r+s)$ time complexity and $O(1)$ memory complexity. The Rand Index computation for the general clustering problem, in contrast, requires the $N$ cluster memberships and has a $O(N)$ complexity in both time and memory.
\end{abstract}


\section{Introduction}

The Rand Index is a classical evaluation metric in Statistics and Machine Learning. Originally proposed for clustering problems (\cite{Rand1971}), it has found applications or adaptations to other tasks such as computer vision(\cite{Unnikrishnan2005AMF, Unnikrishnan_2007, liu2019quantitative}), co-reference resolution (\cite{recasens_hovy_2011}), and change-point detection (\cite{Truong2020, van2020evaluation}). Since the number of data points $N$ can be large depending on the application, efficient algorithms that are linear in $N$ in both time and memory complexity have been proposed (\cite{filkov2004integrating, sundqvist2020adjusting}).

The proposed algorithms in the literature, however, are suited for general purpose clustering. When restricted to change-point detection, the additional structure of the problem studied imposes that the clusters detected must be contiguous. For this scenario, we prove that the Rand Index can be computed in a more efficient manner and the input can be compressed to have just two change-points sets of sizes $s$ and $r$ instead of the full membership of the elements. The algorithm becomes much more efficient the number of change points is much smaller than then number of data points, i.e. $r+s \ll N$.

\section{Original formulation of the Rand Index}

Given two integers $r < s$, we will use the notation $r:s$ for the set $\{r, r+1, \ldots, s\}$. Let $\mathcal{Z}$ be a non-empty set, $N\,\in\, \mathbb{N}$ be the sample size, and $\mathbf{z} = (z_i)_{i=1}^N$ be data samples such that $z_i \, \in \, \mathcal{Z}$. A clustering is defined as a partition $C = \{C_1,\ldots, C_r\}$ of $1:N$, and each partition set is a cluster. This partition is usually learned by applying a statistical or machine learning algorithm in $\mathbf{z}$ to group data points that share similarities.

Denote by $I_{C}(i, j)$ the function that indicates if the samples $z_i$ and $z_j$ are in the same cluster of $C$. Given two clusterings $C_1$ and $C_2$, define

\begin{align*}
&N_{11} = |\{(i,j) \, \in \, (1:N)^2 \, \,| \,\, i < j \mbox{ and } I_{C_1}(i,j) = 1 = I_{C_2}(i,j)\}| \quad,\\
&N_{00} = |\{(i,j) \, \in \, (1:N)^2 \, \,| \,\, i < j \mbox{ and } I_{C_1}(i,j) = 0 = I_{C_2}(i,j)\}| \quad.
\end{align*}

The Rand Index is then defined as

\begin{equation}
RI = \frac{N_{00}+N_{11}}{\binom{N}{2}} \quad.
\end{equation}

The term $N_{11}$ measures how many pairs of indices both clusterings grouped together and the term $N_{00}$ how many pairs are placed in different sets by both clusterings. Therefore, $N_{00}+N_{11}$ measures the total number of agreements between clusterings. Finally, we scale by $\binom{N}{2}$, the total number of pairs. The metric ranges on $[0, 1]$, attaining $1$ if, and only if, the clusterings are identical, and $0$ if they are completely dissimilar.

Write $C_{1} = \{C_{11},\ldots, C_{1r}\}$ and $C_{2} = \{C_{21},\ldots, C_{2s}\}$. A standard algorithm iterates through the partitions to build a $r\times s$ contingency table whose elements are \\ $n_{ij} = |C_{1i}\cap C_{2j}|$, and then computes the Rand Index by the equation

\begin{equation}
RI = 1 - \frac{\left[\frac{1}{2}\left(\sum_{i=1}^r(\sum_{j=1}^s n_{ij})^2 + \sum_{j=1}^s(\sum_{i=1}^r n_{ij})^2\right)  - \sum_{i=1}^r\sum_{j=1}^s(n_{ij})^2\right]}{\binom{N}{2}} \quad.
\label{eq: RI_alternative_eq}
\end{equation}

Therefore, the time complexity of the algorithm is $O(N+rs)$, and its memory
complexity is $O(rs)$ since it needs to store the contingency table. By exploring sparse representations of the contingency table and efficient computations of the symmetric different, efficient algorithms that are $O(N)$ in both time and space were obtained (\cite{filkov2004integrating}, \cite{sundqvist2020adjusting}). As we will show, we can improve both complexities when considering clusters that come from change-point problems.

\section{Efficient Rand Index algorithm for CPD}

Change-point detection is a multidisciplinary field of statistics that provides reliable methodologies for the detection of abrupt changes in time-series. Although its methods are not the focus of this work, we describe a simplified offline formulation of the problem. Consider a sequence $\{Z_{i}\}_{i=1}^N$ of independent random variables where $Z_{i}$ has a cumulative distribution function $F_i$ for all $i \, \in \, 1:N$. Let $C^*$ be the set where a distribution change occurs, that is
    
$$ C^* = \{c^* \in 1:(N-1) | F_{c^*} \neq F_{c^*+1}  \}\quad.$$

$C^*$ is the true change-point set, and its elements are called change-points. Whenever a change-point occurs, the distribution of the data changes, capturing the idea of abrupt change in the process behavior. The random variables between two consecutive change-points have the same distribution so that they can be seen as belonging to the same cluster.

The goal then is to estimate $C^*$ and the related CDFs of each segment. The change-point method outputs a change-point set $C$ that best splits the data in contiguous segments according to some statistical criteria or loss function. It is then usual to study the performance of the methods by comparing the outputted change-point sets between themselves and against the ground truth with respect to some metric, in our case the Rand Index. For an introduction to change-point detection, see \cite{niu2016} and \cite{Truong2020}.

\subsection{Rand Index CPD equation}

Given $C = \{c_1,\ldots, c_r\}$, the sorted change-point set detected, there is a natural identification to a partition of $1:N$. Defining $c_0 = 0$ and $c_{r+1} = N$, the set $C$ can be seen as the clustering

\begin{equation}
\{\{(c_{i}+1):c_{i+1}\}_{i=0}^r\} \quad.
\label{eq:cp_cluster_id}
\end{equation}

A set with $r$ change-points has $r+1$ contiguous clusters, each ending at a change-point. To exemplify, assume that $N = 10$ and $C = \{3, 8\}$. The equivalent clustering is $\{\{1,2,3\}, \{4,5,6,7,8\}, \{9, 10\}\}$. 

We can compute the Rand Index between two change-point sets by comparing their induced clusterings. Since the clusters are contiguous, the Rand Index admits a simplified expression that depends solely on the change-points.

 \begin{theorem}
        Let $C  = \{c_1,c_2,\ldots, c_r\}$ and $C^* = \{c_1^*,c_2^*,\ldots,c_s^*\}$ be sorted change-point sets. Define $c_0 = c_0^* = 0$ and $c_{r+1} = c_{s+1}^* = N$. Identifying the sets with clusterings as in Equation~\ref{eq:cp_cluster_id}, the Rand Index is given by
        
        \begin{equation}
            RI = 1 - \frac{\ouset{i=0}{\sum}{r}\ouset{j=0}{\sum}{s} n_{ij} |c_{i+1} - c_{j+1}^*|}{\binom{N}{2}} \quad, \label{eq: randindex_cpd}
        \end{equation} 
        
        where
        
        $$ n_{ij} =  \max\left(0, \min\left(c_{i+1},c_{j+1}^*\right) - \max\left(c_{i},c_{j}^*\right)\right)\quad.$$
    \end{theorem}

\begin{proof}

For each element $x$ in $1:(N-1)$ define $A_x$ as the set of pairs $(x,y)$ where $x < y$ and in which the clusterings agree. Since these sets are disjoint and contain all and only the pairs that are in agreement, we have that

$$ N_{00} + N_{11} = \sum_{x=1}^{N-1} |A_x| \quad.$$

The restriction $x < y$ avoids double counting.

First, we know that there are a total of $N-x$ pairs $(x,y)$ of the form $x < y$. Let $\phi(x)$ and $\psi(x)$ be the indices of the smallest change-points in $C$ and $C^*$ that are greater or equal to $x$, respectively. Hence, $x\,\in\,(c_{\phi(x)-1}+1):c_{\phi(x)}$ and $x\,\in\,(c_{\psi(x)-1}^*+1):c_{\psi(x)}^*$. The clusterings agree on all pairs $(x, y)$ where $y \, \in \, (x+1):\min(c_{\phi(x)},c_{\psi(x)}^*)$ since they place $x$ and $y$ in a single set. Additionally, they agree on the pairs $(x, w)$ for $w\,\in\, \left(\max\left(c_{\phi(x)},c_{\psi(x)}^*\right)+1\right):N$ since they place $x$ and $w$ on different sets.

\begin{figure}
    \centering
    \includegraphics[width=1\linewidth]{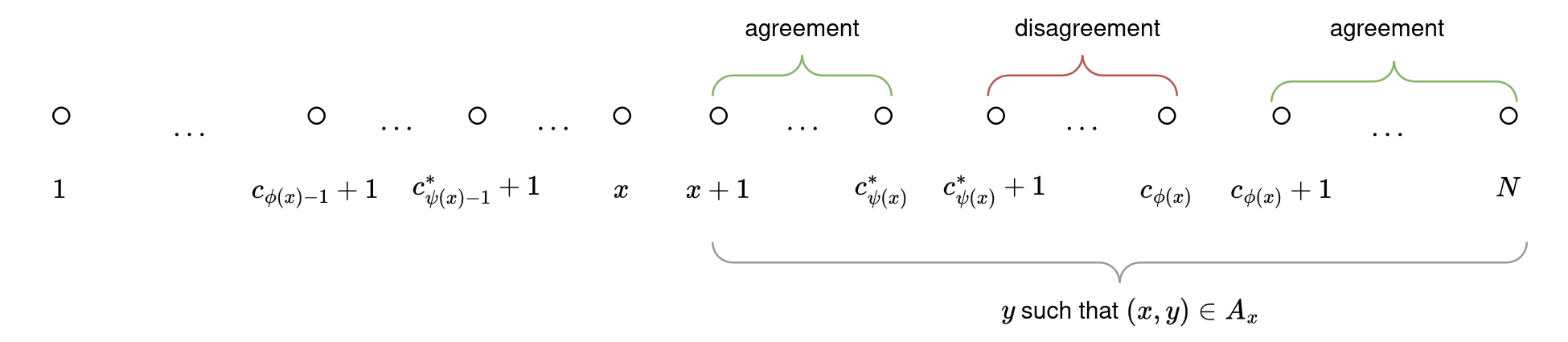}
    \caption{Sketch of agreements and disagreements in $A_x$.}
\end{figure}

From this reasoning, there are a total of $|c_{\phi(x)}-c_{\psi(x)}^*|$ disagreements, so that

$$ |A_x| = (N-x) - |c_{\phi(x)}-c_{\psi(x)}^*| \quad.$$

Substituting back in the original equation

\begin{equation}
\label{eq: RI_L1_connection}
RI = \frac{\sum_{x=1}^{N-1} |A_x|}{\binom{N}{2}} = 1 - \frac{\sum_{x=1}^{N-1} |c_{\phi(x)}-c_{\psi(x)}^*|}{\binom{N}{2}} = 1 - \frac{||c_{\phi} - c^*_\psi||_1}{\binom{N}{2}}\quad,
\end{equation}
where $c_{\phi} = (c_{\phi(x)})_{x=1}^N$ and $c_{\psi}^* = (c_{\psi(x)}^*)_{x=1}^N$. Let $I_{ij} = ((c_{i}+1):c_{i+1}) \cap ((c^*_{j}+1):c^*_{j+1})$. It is easy to show that $(I_{ij})_{i=0,j=0}^{r,s}$ forms a partition for $1:N$. For every element $x$ of $I_{ij}$, we have $\phi(x) = i+1$ and $\psi(x) = j+1$, so that

\begin{align*}
\sum_{x=1}^{N-1} |c_{\phi(x)}-c_{\psi(x)}^*| 
&= \sum_{i=0}^{r}\sum_{j=0}^{s}\sum_{x\,\in\,I_{ij}} |c_{\phi(x)}-c_{\psi(x)}^*|\\
&= \sum_{i=0}^{r}\sum_{j=0}^{s}\sum_{x\,\in\,I_{ij}} |c_{i+1}-c_{j+1}^*|\\
&= \sum_{i=0}^{r}\sum_{j=0}^{s}n_{ij}|c_{i+1} - c_{j+1}^*| \quad,
\end{align*}

where $n_{ij} = |I_{ij}|$. It is simple to show that

$$ n_{ij} = \max\left(0, \min\left(c_{i+1},c_{j+1}^*\right) - \max\left(c_{i},c_{j}^*\right)\right)\quad,$$

which completes the proof.

\end{proof}

    As a minor side note, Equation~\ref{eq: RI_L1_connection} establishes a link between the combinatorial nature of the Rand Index and the $L_1$-geometry of the change-point right-assignment vectors $c_\phi$ and $c_\psi^*$.

\subsection{Rand Index CPD algorithm}

Albeit the summation in Equation~\ref{eq: randindex_cpd} has $rs$ terms, there are at most $r+s$ non-empty $I_{ij}$ intervals, hence at most $r+s$ terms for which $n_{ij} \neq 0$. Indeed, for each interval $((c_{i}+1):c_{i+1})_{i=0}^r$, let $a(i)$ be the number of intervals of $C^*$ that $(c_{i}+1):c_{i+1}$ intersects and $J(i)$ be the highest interval index of $C^*$ that it intersects. Since the $i$-th interval cannot intersect the intervals below $J(i-1)$, we have $a(i) \leq J(i) - (J(i-1) - 1)$. The total number of non-empty intervals is just the sum of $a(i)$, so \\ $\sum_{i=1}^r a(i) \leq \sum_{i=1}^r (J(i) - J(i-1) + 1) \leq r + s$, where we used the fact that $J(r) = s$.

We can efficiently filter the empty cases with the following observations. On one hand, if $i$-th index for the first summation and $j$-th index for the second summation satisfy $c_{i+1} < c_{j+1}^*$, then we must have that $n_{ik} = 0$ for all $k \geq j+1$. This happens because the $\min\left(c_{i+1},c_{k+1}^*\right) - \max\left(c_{i},c_{k}^*\right) = c_{i+1} - c_{k}^* < 0$, hence the set $I_{ij}$ is empty. 
    
On the other hand, if $c_{i+1} \geq c_{j+1}^*$, then $n_{kl} = 0$ for all $k \geq i + 1$ and $l \leq j$. Therefore, for all indices above $i$ in the first summation, we can skip all indices below or equal to $j$ in the second summation. 

The pseudocode below provides an implementation taking these observations into account.

\begin{algorithm}[H]
\caption{Compute Rand Index CPD}
\begin{algorithmic}
\Procedure{RICPD}{$C_1, C_2$}
    \State $r \gets \mbox{size}(C_1) - 1$ \Comment{Note $C_1[0] = 0; C_1[r] = N$}
    \State $s \gets \mbox{size}(C_2) - 1$ \Comment{Note $C_2[0] = 0; C_2[s] = N$}
    \State $d \gets 0$ \Comment{Dissimilarity}
    \State $b \gets 0$ \Comment{initial value for $j$ to skip unnecessary iterations}
    \For{$i\,\in\, 0:(r-1)$}
        \For{$j\,\in\, b:(s-1)$}
            \State $m \gets \min(C_1[i+1], C_2[j+1]) - \max(C_1[i], C_2[j])$
            \State $m \gets \max(0, m)$
            \State $d \gets d + m|C_1[i+1]-C_2[j+1]|$
            \If{$C_1[i+1] < C_2[j+1]$}
                \State $\mathbf{break}$
            \Else
                \State $b \gets j+1$
            \EndIf
            
        \EndFor
    \EndFor
    \State $N \gets C_1[r]$
    \State $RI \gets 1 - \frac{d}{\binom{N}{2}}$
    \State \Return RI
\EndProcedure
\end{algorithmic}
\end{algorithm}

Note that the input of the algorithm is the sorted change-points sets (with the auxiliary change-points $0$ and $N$ at both ends) whose total size is $s + r$. In contrast, the traditional algorithm requires the full partitions whose size is at least $2N$. 

It follows directly that the algorithm uses $O(1)$ auxiliary memory. For the time complexity, given $i\,\in\,0:r$, let $J(i) = \max\{j\,\in\,(0:s)\,\,|\,\, c_{j+1}^* \leq c_{i+1}\}$. The $i$-th iteration does not perform any computations on the indices smaller than $J(i-1)$ since we update the starting location of $j$. Moreover, it breaks on $j = J(i)+1$. Since the number of operations per iteration is constant, the time complexity $T(r, s)$ satisfies

\begin{align*}
    T(r, s) 
    &\leq \beta + \sum_{i=1}^r \alpha(J(i) + 1 - J(i-1))\\
    &\leq \beta + r\alpha + \alpha \sum_{i=1}^r (J(i)- J(i-1)) \quad,\\
    &=\beta + \alpha(r+s) \,\in\, O(r+s)\quad,
\end{align*}
    
    since $J(r) = s$.

\section*{Acknowledgement}

The author thanks Florencia Graciela Leonardi, the advisor of his master's thesis, during which he obtained this minor result while studying change-point problems. The author was supported by a CAPES\footnote{Coordination of Superior Level Staff Improvement, Brazil} fellowship during his master's thesis.

\bibliographystyle{apacite}
\bibliography{bibliography}

\end{document}